\documentclass[submission,copyright,creativecommons]{eptcs}

\providecommand{\event}{GandALF 2026}

\usepackage{iftex}

\ifpdf
  \usepackage{underscore}         
  \usepackage[T1]{fontenc}        
\else
  \usepackage{breakurl}           
\fi

\usepackage{amsmath}
\usepackage{amssymb}
\usepackage{amsthm}
\usepackage{thmtools}
\usepackage[capitalise]{cleveref}
\usepackage{hypcap}
\usepackage{graphicx}
\usepackage{tikz}
\usetikzlibrary{arrows.meta,positioning}

\newtheorem*{question}{Question}

\newcommand{\PSPACE}{\text{PSPACE}}
\newcommand{\PTIME}{\text{PTIME}}
\newcommand{\cocompl}{\text{co}}
\newcommand{\NP}{\text{NP}}

\newcommand{\MP}{\text{MeanPayoff}}

\newcommand{\Parity}{\text{Parity}}

\newcommand{\ModCycle}{\text{Mod-$k$}}

\newcommand{\HamiltonP}{\textsc{Hamiltonian Path}}

\newcommand{\QBF}{\textsc{QBF}}
\newcommand{\QBFcnf}{\textsc{QBF-3CNF}}

\newcommand{\Eve}{\exists}
\newcommand{\Adam}{\forall}

\newcommand{\Ve}{V_{\Eve}}
\newcommand{\Va}{V_{\Adam}}
\newcommand{\A}{\mathcal{A}}
\newcommand{\G}{\mathcal{G}}

\newcommand{\Val}{\text{Val}}

\newcommand{\True}{\textsc{True}}
\newcommand{\False}{\textsc{False}}

\newcommand{\Cond}{W}
\newcommand{\Obj}{\Omega}

\newcommand{\Colors}{C}

\newcommand{\play}{\rho}
\newcommand{\Plays}{\text{Plays}}
\newcommand{\MPlays}{\text{MPlays}}

\newcommand{\str}{\sigma}
\newcommand{\Path}{\pi}

\newcommand{\strE}{\sigma}
\newcommand{\strA}{\tau}

\newcommand{\Pay}{\text{Pay}}

\newcommand{\val}{\text{val}}

\newcommand{\GPay}{\G = (\A, \Pay)}
\newcommand{\GNoPay}{\G = (\A, \Cond)}

\theoremstyle{plain}
\newtheorem{theorem}{Theorem}
\newtheorem{lemma}[theorem]{Lemma}
\newtheorem{proposition}[theorem]{Proposition}

\newtheorem{claim}[theorem]{Claim}

\theoremstyle{definition}
\newtheorem{definition}[theorem]{Definition}

\theoremstyle{remark}
\newtheorem{remark}[theorem]{Remark}

\newenvironment{claimproof}{\begin{proof}[Proof of claim]}{\end{proof}}

\title{Lasso Games}

\author{
  Krzysztof {Żyndul}
  \institute{
    Institute of Informatics, University of Warsaw\\
    Warsaw, Poland
  }
}

\def\titlerunning{Lasso Games}
\def\authorrunning{K. Żyndul}

\begin{document}

\maketitle

\begin{abstract}
  We study games on finite graphs with labelled or weighted edges in which plays end as soon as the first loop is formed. The resulting play has a lasso shape: a path (head) followed by a cycle (loop). We consider two semantics for determining the winner: the \emph{loop semantics}, where only the loop matters, and the \emph{lasso semantics}, where the entire play is taken into account. We show that adopting the lasso semantics may have a dramatic impact on both the complexity of the game and the memory required to play optimally. Mean-payoff games with lasso semantics are no longer positionally determined, and the threshold problem -- deciding whether the value of the game is at least $0$ -- is $\PSPACE$-complete. In contrast, lasso parity games are equivalent to weak parity games, remain positionally determined and are solvable in polynomial time. Furthermore, we show that if a positionally determined winning objective is closed under cyclic permutation and both the objective and its complement are closed under concatenation, then it is equivalent to the parity objective. Finally, we consider mod-$k$ games, where the winner depends on the length of the path modulo $k \geq 2$. We show that these games, in both semantics, are $\PSPACE$-complete. Alongside computational complexity, we analyze the memory required to play optimally. We establish that for any game positionally determined in loop semantics, memory of size $2^n$ is sufficient in lasso semantics. We complement this by showing exponential memory lower bounds for lasso mean-payoff and mod-$k$ games.
\end{abstract}

\section{Introduction}
Games on graphs constitute a mathematical tool that underlies the model-checking paradigm in program verification. Intensive study of these games over the last three decades has revealed deep connections with automata theory, logic, and computational complexity; see the recent monograph by Fijalkow et al.~\cite{NathBook}. Such games are typically of infinite duration, with the winner determined by an infinitary objective (e.g., a parity or mean-payoff objective). The complexity of algorithms for solving these games depends on both the objective and the structure of the underlying graph.

A number of games admit an equivalent finitary variant, where the game can be stopped in finite time~\cite{McNaughton2000, ZimmermannPDH}. In this context, Aminof and Rubin~\cite{Aminof2017} initiated the study of first-cycle games, where the game ends as soon as a vertex is repeated and the winner is determined by the sequence of labels on the \emph{first cycle}.

In this paper, we consider first-cycle games with an alternative semantics. Instead of looking only at the cycle, the winner is determined by the whole \emph{lasso-shaped} path. Our results show that this change of semantics can have a significant impact on the complexity of the game. We study first-cycle games with lasso semantics under mean-payoff, parity, and mod-$k$ objectives. While the first two are classical objectives, mod-$k$ games arise as a game variant of the algorithmic problem of deciding whether a graph contains a cycle whose length is divisible by a fixed integer $k \geq 2$.

\textbf{Our results.} For \textbf{lasso mean-payoff games}, where the value of the game is the arithmetic mean of the weights on the entire lasso-shaped path, we prove that the \emph{threshold problem} -- deciding whether Eve can ensure a payoff of at least $0$ -- is $\PSPACE$-complete. To prove $\PSPACE$-hardness, we give a reduction from the $\QBF$ problem. When restricted to one-player settings, the problem becomes $\NP$-complete for Eve and $\cocompl\NP$-complete for Adam. On arenas without negative-weight cycles, solving one-player Eve games remains $\NP$-hard, while one-player Adam games can be solved in $\PTIME$ for threshold $0$.

Next, we consider \textbf{lasso parity games}, where the winner is determined by the maximal priority on the lasso-shaped path. These games correspond to the weak parity objective. Thus, they are positionally determined and solvable in $\PTIME$. Additionally, we show that any positionally determined objective that is closed under cyclic permutations, and both the objective and its complement are closed under concatenation, then the objective is equivalent to parity objective.

Finally, we consider \textbf{mod-$k$ games}. We prove that in both the loop and lasso semantics, solving these games is $\PSPACE$-complete for every fixed $k \geq 2$.

\textbf{Memory requirements.} We establish that for any game positionally determined in loop semantics (such as mean payoff), memory of size $2^n$ is sufficient in lasso semantics, where $n$ is the number of vertices. We also establish exponential memory bounds of $2^{\Theta(n)}$ for both lasso mean-payoff and mod-$k$ games.

\section{Preliminaries}\label{sec:preliminaries}
For every positive integer $n$, let $[n] = \{1, \dots, n\}$.

An \emph{edge-labelled directed multigraph} is a triple $(V, E, \Colors)$, where $V$ is a finite set of \emph{vertices}, $\Colors$ is a finite set of \emph{labels}, and $E \subseteq V \times \Colors \times V$ is its set of \emph{edges}. Note that there can be many edges between two vertices, labelled by different labels. For $(u, c, v) \in E$, we write $u \xrightarrow{c} v$. We write $E(u) = \{(u, c, v) \in E\}$ for the set of outgoing edges of $u$. A \emph{path} is a finite or infinite sequence $\Path = v_0, c_0, v_1, c_1, \dots$ such that $(v_i, c_i, v_{i + 1}) \in E$; its length $|\Path|$ is its number of edges. A \emph{cycle} is a non-empty path whose first and last vertices are the same, and it is a \emph{simple cycle} if no other vertex repeats.

A game is played on an \emph{arena} $\A = (V, E, \Ve, \Va, v_0, \Colors)$, where $(V, E, \Colors)$ is a finite edge-labelled directed multigraph with no terminal vertices, i.e., $E(v) \neq \emptyset$ for every $v \in V$; $V$ is partitioned into the sets $\Ve$ and $\Va$ controlled by \emph{Eve} and \emph{Adam}; and $v_0$ is the initial vertex.

Adam and Eve play the game by moving a token along the edges of the arena. The token starts at $v_0$. When the token is at a vertex $u$ controlled by player $p \in \{\Eve, \Adam\}$, i.e., $u \in V_p$, player $p$ chooses an outgoing edge $(u, c, v) \in E(u)$ and moves the token to $v$. In this way, the players generate a finite or infinite path $\play = v_0, c_0, v_1, c_1, v_2, \dots$, called a \emph{play}. We write $\Plays(\A)$ for the set of all plays over the arena $\A$. To determine the outcome of a game, we evaluate \emph{maximal plays}. In infinite-duration games, a maximal play is simply an infinite path. In finite-duration games, a maximal play is a finite path that satisfies the stopping condition of the game. We denote by $\MPlays(\A) \subseteq \Plays(\A)$ the set of all maximal plays over an arena $\A$.

A \emph{strategy} for a player determines their next move at any point in the game. Formally, a strategy $\str$ for player $p \in \{\Eve, \Adam\}$ maps each finite play $\play$ ending in $v \in V_p$ to an edge $\str(\play) \in E(v)$. A strategy $\str$ is called \emph{positional} (or \emph{memoryless}) if it depends only on the current vertex rather than the entire history. Such a strategy can be viewed as a function $\str : V_p \to E$ satisfying $\str(v) \in E(v)$ for every $v \in V_p$. A play $\play = v_0, c_0, v_1, c_1, v_2, \dots$ is consistent with a strategy $\str$ of player $p$ if, for all $i$ such that $v_i \in V_p$, we have $\str(v_0, c_0, v_1, \dots, v_i) = (v_i, c_i, v_{i + 1})$. We write $\play_{\strE, \strA}(v_0)$ for the unique play starting at $v_0$ that is consistent with both Eve's strategy $\strE$ and Adam's strategy $\strA$.

A \emph{qualitative game} is a tuple $\G = (\A, \Cond)$, where $\A$ is an arena and $\Cond \subseteq \MPlays(\A)$ is a \emph{winning condition}. A maximal play $\play$ is won by Eve if it belongs to the winning condition. Otherwise, it is won by Adam. Eve's strategy $\strE$ is a \emph{winning strategy} from a vertex $v_0$ if \emph{every} maximal play starting at $v_0$ and consistent with strategy $\strE$ is won by Eve. A winning strategy for Adam is defined analogously.

The \emph{winning region} of player $p \in \{\Eve, \Adam\}$, denoted $W_p$, is the set of all vertices $v \in V$ from which player $p$ has a winning strategy. A game is \emph{determined} if the winning regions of Eve and Adam partition the vertex set $V$ (i.e., $W_\Eve \cup W_\Adam = V$). A game is \emph{positionally determined} if each player has a positional winning strategy from every vertex in their winning region.

For a fixed set of labels $\Colors$, any set $\Obj \subseteq \Colors^* \cup \Colors^{\omega}$ can be considered a \emph{winning objective} that describes families of games with the same type of winning condition. When evaluating an objective, we identify a play with its sequence of edge labels. Given an arena $\A$ and an objective $\Obj$, the induced winning condition is
\[
    \Cond_\Obj = \{ \play \in \MPlays(\A) \mid \play \in \Obj \}.
\]
When the objective $\Obj$ is clear from the context, we often omit the induced condition and simply say that a play $\play$ satisfies objective $\Obj$.

We focus on a specific class of games where play ends as soon as a first cycle is closed. Given an arena $\A = (V, E, \Ve, \Va, v_0, \Colors)$, a \emph{first-cycle} play is of the form
\[
    \play = v_0, c_0, v_1, c_1, v_2, \dots, c_{j - 1}, v_j, c_j, \dots, v_{\ell_{\play}},
\]
where $\ell_{\play}$ is the first index such that there exists $0 \leq j < \ell_{\play}$ with $v_j = v_{\ell_{\play}}$. The play forms a \emph{lasso}, consisting of the \emph{head} $v_0, c_0, v_1, \dots, c_{j - 1}, v_j$, possibly trivial if $j = 0$, and the \emph{loop} $v_j, c_j, \dots, v_{\ell_{\play}}$, which is a simple cycle.

Given a winning objective $\Obj$ and a first-cycle play $\play$, we consider two semantics to determine whether it is won by Eve:
\begin{itemize}
    \item in the \emph{loop semantics} (as studied by Aminof and Rubin~\cite{Aminof2017}), if the sequence of labels on the loop belongs to the objective, i.e., $c_j c_{j + 1} \dots c_{\ell_{\play} - 1} \in \Obj$;

    \item in the \emph{lasso semantics}, if the sequence of labels of the entire play belongs to the objective, i.e., $c_0 c_1 \dots c_{\ell_{\play} - 1} \in \Obj$.
\end{itemize}

We note that all games considered in this paper are determined. Finite perfect-information games are determined by a classical result of Zermelo~\cite{Zermelo1913}, and the infinite games we consider have also been shown to be determined.

We consider the following computational problem for qualitative games. We assume that the winning condition is given by a finite representation.

\begin{question}[Solving a game]
    Given a game $\GNoPay$, decide whether Eve has a winning strategy from $v_0$.
\end{question}

A \emph{quantitative game} is a tuple $\G = (\A, \Pay)$, where $\A$ is an arena and $\Pay : \Colors^* \cup \Colors^{\omega} \to \mathbb{R} \cup \{\pm \infty\}$ is a payoff function. Eve seeks to maximise the payoff, while Adam seeks to minimise it. We define the loop and lasso semantics for quantitative games. Given a payoff function $\Pay$ and a first-cycle play $\play$, its value is defined as follows:
\begin{itemize}
    \item in the \emph{loop semantics} (as studied by Aminof and Rubin~\cite{Aminof2017}), as the value of the sequence of labels on the loop, i.e.,
          $\Pay(c_j c_{j + 1} \dots c_{\ell_{\play} - 1})$;

    \item in the \emph{lasso semantics}, as the value of the sequence of labels of the entire play, i.e., \\
          $\Pay(c_0 c_1 \dots c_{\ell_{\play} - 1})$.
\end{itemize}
Once the semantics is fixed, we use $\Pay(\play)$ for the corresponding value of the play $\play$.

Given strategies $\strE$ for Eve and $\strA$ for Adam, let $\play_{\strE, \strA}(v_0)$ be the unique first-cycle play starting at $v_0$ and consistent with these strategies. The values of the game for Eve and Adam at vertex $v_0$, denoted $\val_{\G}^{\Eve}(v_0)$ and $\val_{\G}^{\Adam}(v_0)$, respectively, are defined as:
\[
    \val_{\G}^{\Eve}(v_0) := \max_{\strE} \min_{\strA} \Pay(\play_{\strE, \strA}(v_0)),
    \qquad
    \val_{\G}^{\Adam}(v_0) := \min_{\strA} \max_{\strE} \Pay(\play_{\strE, \strA}(v_0)).
\]
A game is \emph{determined} if $\val_{\G}^{\Eve}(v) = \val_{\G}^{\Adam}(v)$ for all vertices $v$. In this case, we call this unique value the \emph{value of the game} from $v$, denoted $\val_{\G}(v)$. We will drop the subscript $\G$ when the game is clear from context. A strategy $\strE$ for Eve is \emph{optimal} from $v$ if it guarantees the value of the game, i.e., $\min_{\strA} \Pay(\play_{\strE, \strA}(v)) = \val_\G(v)$; optimal strategies for Adam are defined analogously. A quantitative game is \emph{positionally determined} if both players have optimal positional strategies from every vertex.

Unless stated otherwise, we assume that weights and target values in quantitative games are encoded in binary. Consequently, the size of the game representation is polynomial in the number of vertices, edges, and the maximum bit-length of the labels. We consider the following computational problems for quantitative games.

\begin{question}[Threshold problem]
    Given a game $\GPay$ and a threshold $x \in \mathbb{Q} \cup \{\pm \infty\}$, decide whether $\val_{\G}(v_0) \geq x$.
\end{question}

\begin{question}[Computing the value]
    Given a game $\GPay$, compute $\val_{\G}(v_0)$.
\end{question}

Similar to Theorem~1 in~\cite{Aminof2017}, lasso games with an effectively computable objective or payoff function can be solved in polynomial space. For completeness, we provide the adapted proof below.

\begin{proposition}\label{prop:lasso-complexity}
    If the membership problem for objective $\Obj$ -- given a finite word $w \in \Colors^*$, decide whether $w \in \Obj$ -- is solvable in $\PSPACE$, then solving lasso games with objective $\Obj$ is in $\PSPACE$.

    If the payoff function $\Pay$ is computable in polynomial space and the size of its output is polynomially bounded, then the value of the lasso game with the payoff function $\Pay$ is computable in polynomial space.
\end{proposition}

\begin{proof}
    It suffices to consider quantitative games, since a qualitative objective $\Obj$ can be represented by the Boolean payoff $\Pay_{\Obj}(w) = 1$ if $w \in \Obj$, and $\Pay_{\Obj}(w) = 0$ otherwise.

    Let $n = |V|$. We unfold the arena into a finite game tree and evaluate it by depth-first search. Since a play ends when a vertex is repeated, the tree has depth at most $n + 1$. At Eve's vertices, the algorithm takes the maximum value over the outgoing edges; at Adam's vertices, it takes the minimum. At each leaf, it computes the payoff of the resulting first-cycle play.

    The algorithm stores only the current path and one intermediate value at each depth. Since $\Pay$ is computable in polynomial space and its values have polynomial-size representations, the computation uses polynomial space. Hence, the value of the game can be computed in polynomial space.
\end{proof}

\section{Mean-Payoff}\label{sec:MP}
We consider finite variants of \emph{mean-payoff games}, where the value is the average of the weights traversed before the first cycle closes.

\begin{definition}[Mean-payoff game]
    Let the set of labels (weights) be a finite subset of integers, $\Colors \subseteq \mathbb{Z}$. The mean-payoff function assigns to every non-empty finite label sequence its arithmetic mean:
    \[
        \MP(c_0 \dots c_{\ell - 1}) = \frac{1}{\ell} \sum_{i = 0}^{\ell - 1} c_i.
    \]
\end{definition}

Classical infinite-duration mean-payoff games were first studied by Ehrenfeucht and Mycielski~\cite{Ehrenfeucht1979Positional}, who proved their positional determinacy. The first-cycle variant in loop semantics is equivalent to the classical game and therefore, the threshold problem has the same complexity $\NP \cap \cocompl\NP$~\cite{Zwick1996}. Shifting to lasso semantics makes the problem $\PSPACE$-hard. For the threshold $0$, we can use the \emph{sum} of the weights instead of their average to determine whether the payoff is non-negative.

\begin{theorem}\label{thm:lasso-MP-complexity}
    The threshold problem with fixed threshold $0$ is $\PSPACE$-complete for lasso mean-payoff games. The hardness already holds on arenas without negative-weight cycles.
\end{theorem}

\begin{proof}
    By~\Cref{prop:lasso-complexity}, we can compute the value of a lasso mean-payoff game from a given vertex and thus solve even more general problems in polynomial space.

    To prove hardness, we show a reduction from $\QBFcnf$ problem. The reduction is similar to the one proposed by Aminof and Rubin~\cite[Theorem~1]{Aminof2017} and to the well-known reduction for generalised geography (cf., e.g., \cite[Chapter~19]{papa1994}). Given a quantified Boolean formula
    \[
        \Phi = Q_1 x_1 \, Q_2 x_2 \, \dots \, Q_n x_n \, \phi,
    \]
    where $\phi$ is a propositional formula in 3-CNF, we construct a lasso mean-payoff game $\G_\Phi$ such that Eve can achieve a value of at least $0$ from $v_0$ if and only if the quantified Boolean formula $\Phi$ is true.

    Intuitively, the game proceeds in two phases: a \emph{variable assignment phase}, where players assign values to the variables $x_1, \dots, x_n$, followed by a \emph{verification phase}, where Adam attempts to show that the resulting valuation falsifies the formula. If Adam chooses a clause that is not satisfied by the assignment, all of its literals evaluate to false. No matter which literal Eve chooses, the game does not end immediately. Instead, an additional edge of weight $-1$ is traversed. If the clause is satisfied, Eve can choose a satisfying literal, and the corresponding edge ends the play immediately.

    The variable assignment phase consists of $n$ connected \emph{variable gadgets}, each corresponding to a variable in the formula $\Phi$. The $i$-th gadget begins at a vertex $v_i$ controlled by the player associated with the quantifier $Q_i$. From $v_i$, that player assigns $\True$ or $\False$ to $x_i$ by moving to $t_i$ or $f_i$, respectively, along an edge of weight $1$. Both vertices then connect to the next gadget via an edge of weight $-1$. Thus, every assignment contributes a total weight of $0$. We also give the initial edge $v_0 \to v_1$ weight $-1$ and the edge entering vertex $c$ weight $1$. Consequently, the underlying graph of $\G_\Phi$ has no negative cycles, and the accumulated sum upon reaching the verification phase at $c$ is $0$.

    The verification phase begins at vertex $c$, controlled by Adam. From $c$, Adam selects a clause $C_j$ from the 3-CNF formula $\phi$. Each clause is represented by a vertex with three outgoing edges leading to vertices corresponding to its literals $l_{j_1}, l_{j_2}, l_{j_3}$. This clause vertex is controlled by Eve. If a literal $l_{j_k}$ corresponds to the positive variable $x_i$, then it has an edge leading to $t_i$; if it corresponds to the negated variable $\neg x_i$, then it has an edge leading to $f_i$. An example of this construction is illustrated in~\Cref{fig:lasso-MP-hardness}.

    \begin{figure}[hbt!]
        \centering
        \begin{tikzpicture}[
        circ/.style={circle, draw, minimum size=7mm, inner sep=0pt},
        sqr/.style={rectangle, draw, minimum size=7mm, inner sep=0pt},
        arr/.style={->, >=stealth},
        weight/.style={font=\scriptsize, inner sep=1.2pt, text=black}
    ]
    \def\chainSpread{1.20}
    \def\cPosition{(6.30, 4.25)}
    \def\clauseSpread{3.0}
    \def\literalSpread{0.90}

    \node[circ] (v0) at (0, 0) {$v_0$};
    \node[circ] (v1) at ({1.25 * \chainSpread}, 0) {$v_1$};

    \draw[arr] (v0) -- node[weight, above] {$-1$} (v1);

    \node[circ] (t1) at ({2.25 * \chainSpread}, 1) {$t_1$};
    \node[circ] (f1) at ({2.25 * \chainSpread}, -1) {$f_1$};
    \node[sqr] (v2) at ({3.25 * \chainSpread}, 0) {$v_2$};

    \draw[arr] (v1) -- node[weight, below right] {$1$} (t1);
    \draw[arr] (v1) -- node[weight, above right] {$1$} (f1);
    \draw[arr] (t1) -- node[weight, below left] {$-1$} (v2);
    \draw[arr] (f1) -- node[weight, above left] {$-1$} (v2);

    \node[anchor=north, yshift=-3mm] at (v1.south) {$x_1$};

    \node[circ] (t2) at ({4.25 * \chainSpread}, 1) {$t_2$};
    \node[circ] (f2) at ({4.25 * \chainSpread}, -1) {$f_2$};
    \node[circ] (v3) at ({5.25 * \chainSpread}, 0) {$v_3$};

    \draw[arr] (v2) -- node[weight, below right] {$1$} (t2);
    \draw[arr] (v2) -- node[weight, above right] {$1$} (f2);
    \draw[arr] (t2) -- node[weight, below left] {$-1$} (v3);
    \draw[arr] (f2) -- node[weight, above left] {$-1$} (v3);

    \node[anchor=north, yshift=-3mm] at (v2.south) {$x_2$};

    \node[circ] (t3) at ({6.25 * \chainSpread}, 1) {$t_3$};
    \node[circ] (f3) at ({6.25 * \chainSpread}, -1) {$f_3$};
    \node[sqr] (v4) at ({7.25 * \chainSpread}, 0) {$v_4$};

    \draw[arr] (v3) -- node[weight, below right] {$1$} (t3);
    \draw[arr] (v3) -- node[weight, above right] {$1$} (f3);
    \draw[arr] (t3) -- node[weight, below left] {$-1$} (v4);
    \draw[arr] (f3) -- node[weight, above left] {$-1$} (v4);

    \node[anchor=north, yshift=-3mm] at (v3.south) {$x_3$};

    % ...
    \node[circ] (vn) at ({8.75 * \chainSpread}, 0) {$v_n$};

    \draw[arr] (v4) -- node[midway, fill=white, inner sep=3pt] {$\cdots$} (vn.west);

    \node[circ] (tn) at ({9.75 * \chainSpread}, 1) {$t_n$};
    \node[circ] (fn) at ({9.75 * \chainSpread}, -1) {$f_n$};
    \node[circ] (chainEnd) at ({10.75 * \chainSpread}, 0) {};

    \draw[arr] (vn) -- node[weight, below right] {$1$} (tn);
    \draw[arr] (vn) -- node[weight, above right] {$1$} (fn);
    \draw[arr] (tn) -- node[weight, below left] {$-1$} (chainEnd);
    \draw[arr] (fn) -- node[weight, above left] {$-1$} (chainEnd);

    \node[anchor=north, yshift=-3mm] at (vn.south) {$x_n$};

    \node[sqr] (C) at \cPosition {$c$};
    \node[circ] (C1) at ([xshift=-\clauseSpread cm, yshift=-1.20cm]C) {$C_1$};

    \node[circ] (litx1) at ([xshift=-\literalSpread cm, yshift=-1cm]C1) {$x_1$};
    \node[circ] (litx2) at ([yshift=-1cm]C1) {$x_2$};
    \node[circ] (litnx3) at ([xshift=\literalSpread cm, yshift=-1cm]C1) {$\bar{x}_3$};

    \draw[arr] (C) -- (C1);
    \draw[arr] (C1) -- (litx1);
    \draw[arr] (C1) -- (litx2);
    \draw[arr] (C1) -- (litnx3);

    \node[circ] (C2) at ([yshift=-1.20cm]C) {$C_2$};
    \node[circ] (litnx1) at ([xshift=-\literalSpread cm, yshift=-1cm]C2) {$\bar{x}_1$};
    \node[circ] (litnx2) at ([yshift=-1cm]C2) {$\bar{x}_2$};
    \node[circ] (litc23) at ([xshift=\literalSpread cm, yshift=-1cm]C2) {};

    \draw[arr] (C) -- (C2);
    \draw[arr] (C2) -- (litnx1);
    \draw[arr] (C2) -- (litnx2);
    \draw[arr] (C2) -- (litc23);

    \path (C2.center) -- ([xshift=\clauseSpread cm]C2.center) node[midway] {$\cdots$};

    \node[circ] (Cm) at ([xshift=\clauseSpread cm, yshift=-1.20cm]C) {$C_m$};
    \node[circ] (litcm1) at ([xshift=-\literalSpread cm, yshift=-1cm]Cm) {};
    \node[circ] (litcm2) at ([yshift=-1cm]Cm) {};
    \node[circ] (litcm3) at ([xshift=\literalSpread cm, yshift=-1cm]Cm) {};

    \draw[arr] (C) -- (Cm);
    \draw[arr] (Cm) -- (litcm1);
    \draw[arr] (Cm) -- (litcm2);
    \draw[arr] (Cm) -- (litcm3);

    \draw[arr] (litx1) -- (t1);
    \draw[arr] (litx2) -- (t2);
    \draw[arr] (litnx3) -- (f3);

    \draw[arr] (chainEnd) to[out=90, in=0] node[weight, midway, above] {$1$} (C);
\end{tikzpicture}
        \caption{The game $\G_\Phi$. Square vertices belong to Adam and circular vertices to Eve. Unlabelled edges have weight $0$.}
        \label{fig:lasso-MP-hardness}
    \end{figure}
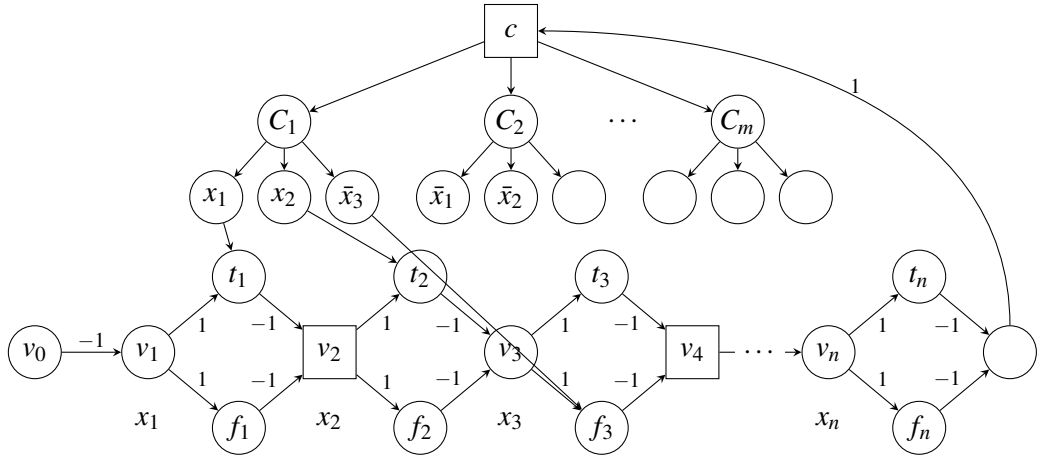

    We establish the correctness of this construction with the following claim.

    \begin{claim}
        For any quantified Boolean formula $\Phi$, the value of the lasso mean-payoff game $\G_\Phi$ from $v_0$ is at least $0$ if and only if $\Phi$ evaluates to true.
    \end{claim}

    \begin{claimproof}
        $\Rightarrow$ We prove the contrapositive. Assume that $\Phi$ is false. Adam has a strategy for assigning the universally quantified variables so that, for any choices made by Eve for the existential variables, the resulting assignment falsifies $\phi$. He can therefore ensure that the assignment formed through the variable gadgets leaves at least one clause $C_j$ unsatisfied. When the play reaches $c$, Adam selects such a clause. All three of its literals evaluate to $\False$, so Eve must select an unsatisfied literal. The corresponding edge leads to a vertex that was \emph{not visited} during the variable assignment phase. The play then traverses one more edge of weight $-1$, and its total sum is $-1$.

        $\Leftarrow$ Assume that $\Phi = Q_1 x_1 \, Q_2 x_2 \, \dots \, Q_n x_n \, \phi$ is true. Eve can choose the values of the existentially quantified variables so that, regardless of Adam's choices for the universally quantified variables, the resulting assignment satisfies $\phi$. In $\G_\Phi$, Eve therefore selects the $\True$ or $\False$ branch in each gadget she controls according to this strategy. After the assignment phase, Adam selects an arbitrary clause $C_j$ at $c$. Because the assignment satisfies $\phi$, the chosen clause has a literal that evaluates to $\True$. Eve selects such a literal at $C_j$. Its edge leads to a vertex \emph{already visited} during the assignment phase and thus closes a cycle. The total sum of the play is $0$.
    \end{claimproof}

    This concludes the hardness proof.
\end{proof}

\subsection{One-player games}

A game is an \emph{Eve game} if every vertex belongs to Eve, and an \emph{Adam game} if every vertex belongs to Adam.

\begin{proposition}\label{prop:one-player-eve-lasso-MP-complexity}
    The threshold problem with fixed threshold $0$ is $\NP$-complete for lasso mean-payoff Eve games. The hardness already holds on arenas without negative-weight cycles.
\end{proposition}

\begin{proof}
    Membership follows by guessing a lasso and checking its sum.

    For hardness, let $H = (V_H, E_H)$ be an instance of $\HamiltonP$ with $n = |V_H|$. We construct a lasso mean-payoff Eve game $\G_H$ as follows. Add a fresh initial vertex $v_0$, an edge of weight $-n$ from $v_0$ to every vertex of $H$, and a self-loop of weight $1$ at every vertex of $H$. Give every edge of $H$ weight $1$, and let every vertex belong to Eve. If $H$ has a Hamiltonian path, Eve follows it and then takes a self-loop, obtaining a total sum $-n + n = 0$. Conversely, a non-negative lasso must traverse at least $n$ edges of weight $1$ after leaving $v_0$. Since the play ends at the first repetition, it must visit all $n$ vertices of $H$, yielding a Hamiltonian path. The arena has no negative cycles.
\end{proof}

\begin{proposition}\label{prop:one-player-adam-lasso-MP-complexity}
    The threshold problem with fixed threshold $0$ is $\cocompl\NP$-complete for lasso mean-payoff Adam games.
\end{proposition}

\begin{proof}
    The complement is in $\NP$, since a negative-sum lasso is a polynomial-size certificate.

    For hardness, we use the same underlying graph as above, let every vertex belong to Adam, assign weight $n - 1$ to each edge leaving $v_0$, and assign weight $-1$ to every other edge. Let $\G_H$ denote the resulting game. By the same counting argument, Adam can produce a negative-sum lasso if and only if $H$ has a Hamiltonian path: such a path followed by a self-loop has sum $(n - 1) - n = -1$, while every negative-sum lasso must traverse all $n$ vertices before its first repetition. Thus, $\val_{\G_H}(v_0) \geq 0$ if and only if $H$ has no Hamiltonian path.
\end{proof}

Restricting the arenas to those without negative cycles does not simplify lasso mean-payoff Eve games. For Adam games, however, the problem becomes polynomial-time solvable.

\begin{lemma}\label{lem:minimal-lasso-ptime}
    Let $G = (V, E)$ be a weighted directed multigraph with no negative cycles, and let $v_0 \in V$. A minimum-sum lasso starting at $v_0$ can be computed in polynomial time.
\end{lemma}

\begin{proof}
    Recall that a \emph{flow network} is a directed graph with a source and a sink, and with a capacity and a cost assigned to each edge. A flow assigns a non-negative amount to every edge without exceeding its capacity. Except at the source and sink, the amount entering each vertex must equal the amount leaving it. The value of the flow is the total amount sent from the source to the sink, and its cost is the sum of the edge costs multiplied by the amounts sent through them. The \emph{minimum-cost flow problem} asks for a flow of a given value and minimum cost.

    For each vertex $v \in V$, we construct a flow network $H_v$ in which $v$ is the only vertex shared by the head and the loop. Split every vertex $u$ into $u_{\mathrm{in}}$ and $u_{\mathrm{out}}$. Except for $v$, connect $u_{\mathrm{in}}$ to $u_{\mathrm{out}}$ by an edge of cost $0$ and capacity $1$. For every edge $(u, c, w) \in E$, add an edge $u_{\mathrm{out}} \to w_{\mathrm{in}}$ of cost $c$ and capacity $1$. Add a source $s$ with edges of cost $0$ and capacity $1$ to $v_{0, \mathrm{in}}$ and $v_{\mathrm{out}}$. Take $v_{\mathrm{in}}$ as the sink and require a flow of value $2$. This construction is illustrated in~\Cref{fig:lasso-MP-flow}.

    An integral flow of value $2$ consists of a path from $v_0$ to $v$ and a path from $v$ back to $v$, possibly together with additional cycles. Such cycles can be removed without increasing the cost because $G$ has no negative cycles. The internal capacities ensure that the two paths are vertex-disjoint except at $v$. Conversely, every lasso whose head and loop share only the vertex $v$ induces a flow in $H_v$ of value $2$ with the same cost. Computing a minimum-cost flow of value $2$ can be done in polynomial time~\cite{Ford1962flows}.

    Thus, solving $H_v$ for every $v \in V$ and taking the cheapest resulting lasso yields a minimum-sum lasso in polynomial time.

    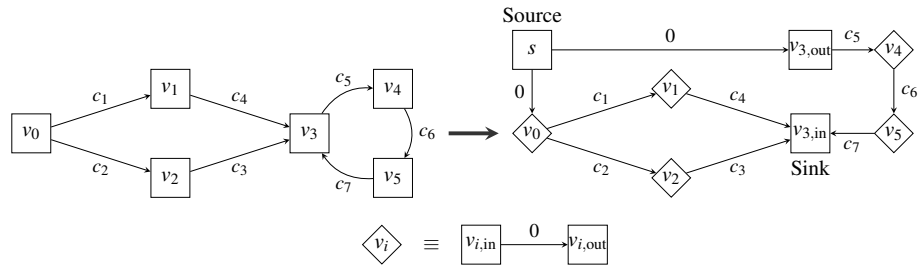
\begin{figure}[!ht]
        \centering
        \resizebox{0.8\textwidth}{!}{\usetikzlibrary{shapes.geometric}

\begin{tikzpicture}[
        transform shape,
        sqr/.style={rectangle, draw, minimum size=7mm, inner sep=0pt},
        dia/.style={diamond, draw, minimum size=7mm, inner sep=0pt},
        arr/.style={->, >=stealth},
        weight/.style={font=\small, midway}
    ]

    \begin{scope}
        \node[sqr] (v0) at (0, 0) {$v_0$};
        \node[sqr] (v1) at (2.5, 0.8) {$v_1$};
        \node[sqr] (v2) at (2.5, -0.8) {$v_2$};
        \node[sqr] (v3) at (5, 0) {$v_3$};
        \node[sqr] (v4) at (6.5, 0.8) {$v_4$};
        \node[sqr] (v5) at (6.5, -0.8) {$v_5$};

        \draw[arr] (v0) -- node[weight, above] {$c_1$} (v1);
        \draw[arr] (v0) -- node[weight, below] {$c_2$} (v2);
        \draw[arr] (v2) -- node[weight, below] {$c_3$} (v3);
        \draw[arr] (v1) -- node[weight, above] {$c_4$} (v3);

        \draw[arr, bend left] (v3) to node[weight, above] {$c_5$} (v4);
        \draw[arr, bend left] (v4) to node[weight, right] {$c_6$} (v5);
        \draw[arr, bend left] (v5) to node[weight, below] {$c_7$} (v3);
    \end{scope}

    \draw[->, >=stealth, line width=2pt, black!80] (7.5, 0) -- (8.4, 0);

    \begin{scope}[xshift=9cm]
        \node[sqr, label=above:{Source}] (s) at (0, 1.5) {$s$};
        \node[dia] (v0b) at (0, 0) {$v_0$};
        \node[dia] (v1b) at (2.5, 0.8) {$v_1$};
        \node[dia] (v2b) at (2.5, -0.8) {$v_2$};
        \node[sqr, label=below:{Sink}] (v3inb) at (5, 0) {$v_{3, \mathrm{in}}$};
        \node[sqr] (v3outb) at (5, 1.5) {$v_{3, \mathrm{out}}$};
        \node[dia] (v4b) at (6.5, 1.5) {$v_4$};

        \node[dia] (v5b) at (6.5, 0) {$v_5$};
        \draw[arr] (s) -- node[weight, left] {$0$} (v0b);
        \draw[arr] (s) -- node[weight, above] {$0$} (v3outb);

        \draw[arr] (v0b) -- node[weight, above] {$c_1$} (v1b);
        \draw[arr] (v0b) -- node[weight, below] {$c_2$} (v2b);
        \draw[arr] (v2b) -- node[weight, below] {$c_3$} (v3inb);
        \draw[arr] (v1b) -- node[weight, above] {$c_4$} (v3inb);

        \draw[arr] (v3outb) -- node[weight, above] {$c_5$} (v4b);
        \draw[arr] (v4b) -- node[weight, right] {$c_6$} (v5b);
        \draw[arr] (v5b) -- node[weight, below] {$c_7$} (v3inb);
    \end{scope}

    \begin{scope}[xshift=6.3cm, yshift=-2cm]
        \node[dia] (vi) at (0, 0) {$v_i$};
        \node[right=0.25cm of vi] (equiv) {$\equiv$};

        \node[sqr, right=0.25cm of equiv] (viin) {$v_{i, \mathrm{in}}$};
        \node[sqr, right=1.2cm of viin] (viout) {$v_{i, \mathrm{out}}$};

        \draw[arr] (viin) -- node[weight, above] {$0$} (viout);

    \end{scope}
\end{tikzpicture}}
        \caption{The flow network $H_v$ with pivot $v_3$.}
        \label{fig:lasso-MP-flow}
    \end{figure}
\end{proof}

\begin{restatable}{proposition}{propAdamMPOptimistic}\label{prop:Adam-lasso-MP-optimistic}
    The threshold problem with fixed threshold $0$ for lasso mean-payoff Adam games on arenas with no negative cycles can be solved in $\PTIME$.
\end{restatable}

\begin{proof}
    A lasso has mean payoff below $0$ if and only if its total sum is negative. By~\Cref{lem:minimal-lasso-ptime}, we can compute the minimum sum $S$ of a lasso from $v_0$ in polynomial time. Adam can force a value below $0$ if $S < 0$. Hence, the threshold problem is in $\PTIME$.
\end{proof}

\begin{remark}
    The algorithm minimises the total sum, which suffices for threshold $0$, but it does not compute the value of the game. It nevertheless applies to every non-positive threshold $x \leq 0$. Let $G_x$ be obtained from $G$ by replacing every edge $(u, c, v)$ with $(u, c - x, v)$. A lasso $\play$ satisfies $\MP(\play) < x$ in $G$ if and only if it has negative total sum in $G_x$. Since $x \leq 0$, this shift cannot create negative cycles, so~\Cref{lem:minimal-lasso-ptime} applies.
\end{remark}

\section{Parity}\label{sec:Parity}
We consider finite variants of \emph{parity games} in which the largest priority seen before the first cycle closes determines the winner.

\begin{definition}[Finite parity objective]
    Let the set of labels (priorities) be a finite subset of natural numbers, $\Colors \subseteq \mathbb{N}$. The \emph{finite parity objective} is the set of finite sequences where the maximum priority is even:
    \[
        \Parity = \left\{c_0 \dots c_{\ell - 1} \in \Colors^+ \; \middle|\; \max_{0 \leq i < \ell} c_i \text{ is even}\right\}.
    \]
\end{definition}

In classical infinite-duration parity games, Eve wins if the maximum priority seen \emph{infinitely often} is even. These games are positionally determined~\cite{Emerson1991,Mostowski1991}. The loop variant is equivalent to classical parity games and lies in $\NP \cap \cocompl\NP$.

Lasso parity games are equivalent to \emph{weak parity games}, in which the winner is determined by the highest priority occurring \emph{at any point} during the play. L\"oding and Thomas~\cite[Lemma~2]{weakParity} showed that weak parity games are solvable in $\PTIME$. Their algorithm, which recursively computes attractors to the highest priorities in a subarena, also solves the lasso variant. Attractor strategies are positional, thus, lasso parity games are positionally determined.

\subsection{Generalised parity objectives}\label{sec:generalize}

Positional determinacy is a strong property satisfied by relatively few objectives. We show that a significant subclass of such objectives is equivalent to the parity objective and therefore shares its tractability.

\begin{definition}[Generalised parity objective]\label{def:gen-parity}
    Let $\Colors$ be finite. An objective $\Obj \subseteq \Colors^+$ is a \emph{generalised parity objective} of order $n$ if there exists a mapping $h : \Colors \to \{0, \dots, n\}$ such that, for every finite sequence $c_0 \dots c_{\ell - 1} \in \Colors^+$, we have $c_0 \dots c_{\ell - 1} \in \Obj$ if and only if $\max_{0 \leq i < \ell} h(c_i)$ is even.
\end{definition}

We begin by formally defining closure properties. Let $\Colors$ be a finite set of labels and let $\Obj \subseteq \Colors^+$. We say $\Obj$ is:
\begin{itemize}
    \item \emph{closed under cyclic permutations} if for all $u, v \in \Colors^*$, $uv \in \Obj \implies vu \in \Obj$.
    \item \emph{closed under concatenation} if for all $u, v \in \Colors^*$, $(u \in \Obj \land v \in \Obj) \implies uv \in \Obj$.
\end{itemize}

\begin{proposition}\label{prop:generalised-weak}
    If an objective $\Obj$ is closed under cyclic permutations, both $\Obj$ and $\overline{\Obj}$ are closed under concatenation, and every game with objective $\Obj$ is positionally determined, then $\Obj$ is a generalised parity objective.
\end{proposition}

In the rest of this section, we prove \Cref{prop:generalised-weak}. We closely follow the proof by Colcombet and Niwiński~\cite{niwek2006}. However, we cannot apply their result directly as it is stated for infinite-duration games and relies on properties of infinite plays. Instead, we adapt their argument to the finite lasso setting. Assume the hypotheses of the proposition and let
\[
    \Colors_\Eve = \Colors \cap \Obj,
    \qquad
    \Colors_\Adam = \Colors \setminus \Colors_\Eve.
\]

\begin{restatable}{lemma}{lemmaDominating}\label{lemma:dominating}
    If $a \in \Colors_\Eve$ and $B \subseteq \Colors$ satisfy $aB \subseteq \Obj$, then $aB^* \subseteq \Obj$.
\end{restatable}

\begin{proof}
    We prove this by induction on the length of the words in $B^+$. The base cases follow directly from the assumption. For the induction step, fix $n \geq 1$, assume that $aw' \in \Obj$ for all $w' \in B^n$, and let $w = w'b \in B^{n + 1}$.

    Assume for the sake of contradiction that $aw'b \notin \Obj$. Consider the game depicted in~\Cref{fig:Parity-dominating}, with winning condition $\Obj$. Since $aw'b \notin \Obj$, Adam has a winning strategy. Because the game is positionally determined, he has a positional winning strategy $\strA$. Without loss of generality, assume that $\strA(v_2) = (v_2, w', v_{1, L})$. Eve can then win by playing $\strE(v_1) = (v_1, b, v_{1, R})$. The resulting play has label sequence $baw'a$. Since $\Obj$ is closed under concatenation, $aw' \in \Obj$ and $ab \in \Obj$ imply $aw'ab \in \Obj$. Furthermore, because $\Obj$ is closed under cyclic permutations, $aw'ab \in \Obj \implies baw'a \in \Obj$. Thus, the play is winning for Eve, contradicting Adam's winning strategy.
\end{proof}

\begin{figure}[htbp]
    \centering
    \begin{tikzpicture}[
        circ/.style={circle, draw, minimum size=7mm, inner sep=0pt},
        sqr/.style={rectangle, draw, minimum size=7mm, inner sep=0pt},
        arr/.style={->, >=stealth}, weight/.style={font=\small, inner sep=2pt}
    ]
    \def\xsep{2}
    \def\ysep{1}

    \node[circ] (v1) at (0, 0) {$v_1$};
    \node[sqr] (v1R) at (\xsep, \ysep) {$v_{1, R}$};
    \node[sqr] (v1L) at (\xsep, -\ysep) {$v_{1, L}$};
    \node[sqr] (v2) at (2 * \xsep, 0) {$v_2$};

    \draw[arr] (v1) -- node[weight, above left] {$b$} (v1R);
    \draw[arr] (v1) -- node[weight, below left] {$b$} (v1L);
    \draw[arr] (v1R) to[bend left=18] node[weight, above right] {$a$} (v2);
    \draw[arr] (v2) to[bend left=18] node[weight, left] {$w'$} (v1R);
    \draw[arr] (v1L) to[bend right=18] node[weight, below right] {$a$} (v2);
    \draw[arr] (v2) to[bend right=18] node[weight, left] {$w'$} (v1L);
\end{tikzpicture}
    \caption{Arena for the proof of~\Cref{lemma:dominating}, with initial vertex $v_1$.}
    \label{fig:Parity-dominating}
\end{figure}
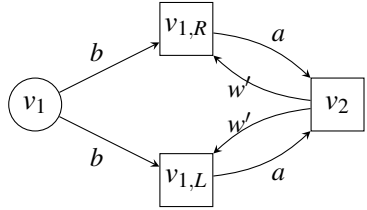

For $a \in \Colors_{\Eve}$, let $g(a) = \{ b \in \Colors_{\Adam} \mid ab \in \Obj \}$. We define the relation $\sqsubseteq$ on $\Colors_{\Eve}$ by
\[
    a \sqsubseteq a' \iff g(a) \subseteq g(a').
\]
We define $\sim$ as the equivalence relation associated with $\sqsubseteq$, and $\sqsubset$ as the strict ordering.

\begin{restatable}{lemma}{lemmaPreorder}\label{lemma:preorder}
    The relation $\sqsubseteq$ is a total preorder.
\end{restatable}

\begin{proof}
    Suppose $a, a' \in \Colors_{\Eve}$ are incomparable. Then there exist $b \in g(a) \setminus g(a')$ and $b' \in g(a') \setminus g(a)$. By definition, $ab \in \Obj$ and $a'b' \in \Obj$. By assumption, $\Obj$ is closed under concatenation, so $aba'b' \in \Obj$. Similarly, because $b \notin g(a')$, we have $a'b \in \overline{\Obj}$. By cyclic permutation, $ba' \in \overline{\Obj}$. Using the same argument, $b'a \in \overline{\Obj}$. Because $\overline{\Obj}$ is also closed under concatenation, $(ba')(b'a) = ba'b'a \in \overline{\Obj}$. However, cyclically permuting $aba'b' \in \Obj$ yields $ba'b'a \in \Obj$, a contradiction.
\end{proof}

\begin{restatable}{lemma}{lemmaPartition}\label{lemma:partition}
    There exist an integer $n$ and $2n$ subsets of $\Colors$, $B_1, \dots, B_{2n}$, such that:
    \begin{enumerate}
        \item The sets $B_i$ form a partition of $\Colors$.
        \item For $k \in [n]$, $B_{2k} \subseteq \Obj$.
        \item For $k \in [n]$, $B_{2k - 1} \subseteq \overline{\Obj}$.
        \item For $k \in [n]$ and $l \leq 2k$, $B_{2k} B_l \subseteq \Obj$.
        \item For $k \in [n]$ and $l \leq 2k - 1$, $B_{2k - 1} B_l \subseteq \overline{\Obj}$.
    \end{enumerate}
\end{restatable}

\begin{proof}
    Let $C_1 \sqsubset \cdots \sqsubset C_p$ be the equivalence classes induced by $\sqsubseteq$, and write $g(C_i)$ for $g(c)$, where $c$ is any element of $C_i$. Let $n = p + 1$, and define the sets $B_i$ as follows:
    \begin{align*}
        B_1        & = g(C_1),                                                                       \\
        B_{2k}     & = C_k,                              &  & \text{for } k \in \{1, \dots, p\},     \\
        B_{2k + 1} & = g(C_{k + 1}) \setminus g(C_k),    &  & \text{for } k \in \{1, \dots, p - 1\}, \\
        B_{2p + 1} & = \Colors_{\Adam} \setminus g(C_p),                                             \\
        B_{2p + 2} & = \emptyset.
    \end{align*}
    By construction, the sets $B_i$ for even $i$ partition $\Colors_{\Eve}$, and the sets $B_i$ for odd $i$ partition $\Colors_{\Adam}$.

    \emph{Property 4.} Consider $k \in [p]$ and $l \leq 2k$. If $l$ is even, then $B_{2k} \subseteq \Obj$ and $B_l \subseteq \Obj$. Since $\Obj$ is closed under concatenation, $B_{2k}B_l \subseteq \Obj$. If $l$ is odd, then by definition $B_l \subseteq g(B_{l + 1})$, and thus $B_l \subseteq g(B_{2k})$ since $B_{l + 1} \sqsubseteq B_{2k}$. By the definition of $g$, we have $B_{2k}B_l \subseteq \Obj$. Finally, for $k = p + 1$, $B_{2k} = \emptyset$, so $B_{2k}B_l \subseteq \Obj$ for every $l$.

    \emph{Property 5.} Let $k \in [n]$ and $l \leq 2k - 1$. If $l$ is odd, then both $B_{2k - 1} \subseteq \overline{\Obj}$ and $B_l \subseteq \overline{\Obj}$. Since $\overline{\Obj}$ is closed under concatenation, $B_{2k - 1}B_l \subseteq \overline{\Obj}$. Suppose that $l$ is even; then $l \leq 2k - 2$. By definition, $B_l = C_j$ for some $j \leq k - 1$, which implies $B_l \sqsubseteq C_{k - 1}$. Furthermore, for any $b \in B_{2k - 1}$, we have $b \notin g(C_{k - 1})$. Therefore, for any $a \in B_l$ and $b \in B_{2k - 1}$, we have $b \notin g(a)$. Consequently, $ab \in \overline{\Obj}$.
\end{proof}

\begin{proof}[Proof of \Cref{prop:generalised-weak}]
    Let $B_1, \dots, B_{2n}$ be given by~\Cref{lemma:partition}, and define $h(c) = i$ whenever $c \in B_i$. For a word $w$, choose a letter $c \in B_i$ of maximal priority and cyclically permute $w$ into $cw'$. Let $D_i = \bigcup_{j \leq i} B_j$; then $w' \in D_i^*$.

    If $i$ is even, \Cref{lemma:partition} gives $c \in \Obj$ and $cD_i \subseteq \Obj$, so \Cref{lemma:dominating} implies $cw' \in \Obj$. If $i$ is odd, the symmetric lemma for $\overline{\Obj}$ yields $cw' \in \overline{\Obj}$. By closure under cyclic permutations, $w \in \Obj$ in the even case and $w \in \overline{\Obj}$ in the odd case, proving that $w \in \Obj$ exactly when $\max h(w)$ is even.
\end{proof}

\begin{remark}\label{rem:generalised-complexity}
    Applying $h$ to the arena labels reduces generalised parity games to parity games, so lasso games with such objectives are solvable in $\PTIME$. Moreover, $h$ depends only on membership of one- and two-letter words. Therefore, if membership can be tested in polynomial time, then $h$ can also be computed in polynomial time.
\end{remark}

\section{Mod-\texorpdfstring{$k$}{k} Objective}\label{sec:Mod-k}
We consider games in which Eve wins if the loop (in loop semantics) or the entire play (in lasso semantics) has length divisible by $k$. Both variants are $\PSPACE$-complete for every fixed $k \geq 2$. In the one-player setting, deciding whether a directed graph contains a simple cycle of odd length~\cite{Thomassen1985} or even length~\cite{Robertson1999,McCuaig2004} is in $\PTIME$. For every fixed $k > 2$ and every fixed residue $r \in \{1, \ldots, k - 1\}$, deciding whether a directed graph contains a simple cycle whose length is congruent to $r$ modulo $k$ is $\NP$-complete~\cite[Theorem~1]{Arkin1991ModularityOC}, whereas the complexity of deciding whether it contains a simple cycle whose length is divisible by $k$ remains open.

\begin{definition}[$\ModCycle$ objective]
    Let $\Colors$ be a singleton and fix $k \geq 2$. The \emph{mod-$k$ objective} is the set of all finite sequences whose length $\ell$ is a multiple of $k$:
    \[
        \ModCycle = \{w \in \Colors^* \mid |w| \equiv 0 \pmod{k}\}.
    \]
\end{definition}

We give a single construction that establishes hardness in both semantics and for every fixed $k \geq 2$. In the figures and the construction below, an edge labelled $\ell$ represents a path of length $\ell$.

\begin{proposition}\label{prop:mod-k-complexity}\label{prop:loop-Even-complexity}\label{prop:lasso-EvenLen-complexity}
    For every fixed $k \geq 2$, solving mod-$k$ games is $\PSPACE$-complete in both loop and lasso semantics.
\end{proposition}

\begin{proof}
    Membership follows from~\cite[Theorem~1]{Aminof2017} and from~\Cref{prop:lasso-complexity}.

    For hardness, we reduce from $\QBF$. Given $\Phi = Q_1 x_1 \cdots Q_n x_n \, \phi$, with $\phi$ in 3-CNF, we construct $\G_\Phi$ such that Eve wins $\G_\Phi$ if and only if $\Phi$ is true. The variable assignment phase consists of $n$ connected \emph{variable gadgets} (see~\Cref{fig:mod-k-gadget}), each corresponding to one variable of the formula $\Phi$. The $i$-th gadget begins at a vertex $v_i$ controlled by the player associated with the quantifier $Q_i$. From $v_i$, that player assigns $\True$ or $\False$ to $x_i$ by moving to $t_i$ or $f_i$, respectively. We also add cross-edges $t_i \xrightarrow{1} f_i$ and $f_i \xrightarrow{k - 1} t_i$, and backtracking edges $v'_i \xrightarrow{k - 1} t_i$ and $v'_i \xrightarrow{1} f_i$. The variable gadgets enforce correct play by both players.

    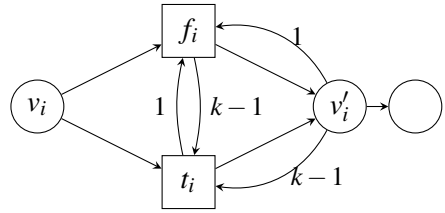
\begin{figure}[htbp]
        \centering
        \begin{tikzpicture}[
        circ/.style={circle, draw, minimum size=7mm, inner sep=0pt},
        sqr/.style={rectangle, draw, minimum size=7mm, inner sep=0pt},
        arr/.style={->, >=stealth}, weight/.style={font=\small}
    ]
    \def\s{1}

    \node[circ] (vi) at (0, {1 * \s}) {$v_i$};
    \node[sqr] (ti) at ({2 * \s}, 0) {$t_i$};
    \node[sqr] (fi) at ({2 * \s}, {2 * \s}) {$f_i$};
    \node[circ] (vip) at ({4 * \s}, {1 * \s}) {$v'_i$};
    \node[circ] (next) at ({5 * \s}, {1 * \s}) {};

    \draw[arr] (vi) -- (ti);
    \draw[arr] (vi) -- (fi);
    \draw[arr] (ti) -- (vip);
    \draw[arr] (fi) -- (vip);
    \draw[arr] (vip) -- (next);

    \draw[arr] (ti) to[bend left=12] node[weight, left] {$1$} (fi);
    \draw[arr] (fi) to[bend left=12] node[weight, right] {$k - 1$} (ti);

    \draw[arr] (vip) to[bend left=35] node[weight, right] {$k - 1$} (ti);
    \draw[arr] (vip) to[bend right=35] node[weight, right] {$1$} (fi);
\end{tikzpicture}
        \caption{Mod-$k$ variable gadget; unlabelled edges are paths of length $k$.}
        \label{fig:mod-k-gadget}\label{fig:even-length-gadget}
    \end{figure}

    \begin{claim}\label{clm:mod-k-gadget}
        Any deviation from correct play in the assignment phase allows the opponent to win in both loop and lasso semantics.
    \end{claim}

    \begin{claimproof}
        \emph{Adam's deviations.}
        Suppose Adam takes a cross-edge. If he takes the opposite cross-edge next, he closes a cycle of length $1 + (k - 1) = k$. Otherwise, he moves to $v'_i$ and Eve returns to the originally chosen vertex, closing a cycle of length $1 + k + (k - 1) = 2k$.

        \emph{Eve's deviations.}
        Suppose the chosen vertex was $t_i$ and Adam moves to $v'_i$ according to correct play. If Eve returns to $t_i$, she closes a cycle of length $k + (k - 1) = 2k - 1$. If she moves to $f_i$, Adam moves back to $v'_i$, closing a cycle of length $k + 1$.

        Finally, every vertex reached during correct play is reached after a multiple of $k$ edges. In all the cases above, the repeated vertex is one of these vertices. Hence the head has length divisible by $k$, so the loop and the whole lasso have the same length modulo $k$.
    \end{claimproof}

    The verification phase begins at an Adam vertex $c$, from which he selects a clause $C_j$. All clause vertices belong to Eve, and each has outgoing edges to its literals: $x_i$ leads to $t_i$ and $\neg x_i$ leads to $f_i$. If Eve chooses a true literal, she returns to a visited vertex and closes a loop of length divisible by $k$. If she chooses a false literal, the corresponding vertex $d_i \in \{t_i, f_i\}$ was not visited; Adam takes the cross-edge to the visited opposite vertex $\bar{d}_i$, closing a loop whose length is congruent to $1$ or $k - 1$ modulo $k$. In either case, the repeated vertex was first reached after a multiple of $k$ edges, so the lasso has the same residue as the loop. We establish the correctness of this construction with the following claim.

    \begin{claim}
        Eve wins $\G_\Phi$ from $v_0$, in either semantics, if and only if $\Phi$ is true.
    \end{claim}

    \begin{claimproof}
        $\Leftarrow$
        Assume $\Phi$ is true. Eve follows her winning assignment strategy. If Adam deviates in a gadget, she wins by~\Cref{clm:mod-k-gadget}. Otherwise, every clause he chooses at $c$ contains a true literal. Eve selects it, returning to a visited vertex and ending with both loop and lasso lengths divisible by $k$.

        $\Rightarrow$
        Assume $\Phi$ is false. Adam follows his winning assignment strategy. If Eve deviates in a gadget, he wins by~\Cref{clm:mod-k-gadget}. Otherwise, he selects an unsatisfied clause at $c$. Every literal leads to an unvisited vertex $d_i$. Adam takes the cross-edge from $d_i$ to $\bar{d}_i$, and neither the loop nor the lasso has length divisible by $k$.
    \end{claimproof}
    Since $k$ is fixed, the construction has polynomial size. This concludes the hardness proof.
\end{proof}
\begin{figure}[htbp]
    \centering
    \begin{tikzpicture}[
        circ/.style={circle, draw, minimum size=7mm, inner sep=0pt},
        sqr/.style={rectangle, draw, minimum size=7mm, inner sep=0pt},
        arr/.style={->, >=stealth},
        weight/.style={font=\scriptsize, inner sep=1.2pt}
    ]
    \def\chainSpread{1.10}
    \def\cPosition{(7.59, 4.30)}
    \def\clauseSpread{3.00}
    \def\literalSpread{0.90}

    \node[circ] (v0) at (0, 0) {$v_0$};
    \node[circ] (v1) at (\chainSpread, 0) {$v_1$};

    \draw[arr] (v0) -- (v1);

    \node[sqr] (t1) at ({2 * \chainSpread}, 0.90) {$t_1$};
    \node[sqr] (f1) at ({2 * \chainSpread}, -0.90) {$f_1$};
    \node[circ] (a1) at ({3 * \chainSpread}, 0) {$v'_1$};
    \node[sqr] (v2) at ({4 * \chainSpread}, 0) {$v_2$};

    \draw[arr] (v1) -- (t1);
    \draw[arr] (v1) -- (f1);
    \draw[arr] (t1) -- (a1);
    \draw[arr] (f1) -- (a1);
    \draw[arr] (a1) -- (v2);

    \draw[arr] (t1) to[bend left=10] node[weight, right] {$1$} (f1);
    \draw[arr] (f1) to[bend left=10] node[weight, left, xshift=1pt] {$k - 1$} (t1);

    \draw[arr] (a1) to[bend right=35] node[weight, above right] {$k - 1$} (t1);
    \draw[arr] (a1) to[bend left=35] node[weight, below right] {$1$} (f1);

    \node[anchor=north, yshift=-3mm] at (v1.south) {$x_1$};

    \node[sqr] (t2) at ({5 * \chainSpread}, 0.90) {$t_2$};
    \node[sqr] (f2) at ({5 * \chainSpread}, -0.90) {$f_2$};
    \node[circ] (a2) at ({6 * \chainSpread}, 0) {$v'_2$};
    \node[circ] (v3) at ({7 * \chainSpread}, 0) {$v_3$};

    \draw[arr] (v2) -- (t2);
    \draw[arr] (v2) -- (f2);
    \draw[arr] (t2) -- (a2);
    \draw[arr] (f2) -- (a2);
    \draw[arr] (a2) -- (v3);

    \draw[arr] (t2) to[bend left=10] node[weight, right] {$1$} (f2);
    \draw[arr] (f2) to[bend left=10] node[weight, left, xshift=1pt] {$k - 1$} (t2);

    \draw[arr] (a2) to[bend right=35] node[weight, above right] {$k - 1$} (t2);
    \draw[arr] (a2) to[bend left=35] node[weight, below right] {$1$} (f2);

    \node[anchor=north, yshift=-3mm] at (v2.south) {$x_2$};

    \node[sqr] (t3) at ({8 * \chainSpread}, 0.90) {$t_3$};
    \node[sqr] (f3) at ({8 * \chainSpread}, -0.90) {$f_3$};
    \node[circ] (a3) at ({9 * \chainSpread}, 0) {$v'_3$};
    \node[sqr] (v4) at ({10 * \chainSpread}, 0) {$v_4$};

    \draw[arr] (v3) -- (t3);
    \draw[arr] (v3) -- (f3);
    \draw[arr] (t3) -- (a3);
    \draw[arr] (f3) -- (a3);
    \draw[arr] (a3) -- (v4);

    \draw[arr] (t3) to[bend left=10] node[weight, right] {$1$} (f3);
    \draw[arr] (f3) to[bend left=10] node[weight, left, xshift=1pt] {$k - 1$} (t3);

    \draw[arr] (a3) to[bend right=35] node[weight, above right] {$k - 1$} (t3);
    \draw[arr] (a3) to[bend left=35] node[weight, below right] {$1$} (f3);

    \node[anchor=north, yshift=-3mm] at (v3.south) {$x_3$};

    % ...
    \node[circ] (vn) at ({11.8 * \chainSpread}, 0) {$v_n$};

    \draw[arr] (v4) -- node[midway, fill=white, inner sep=3pt] {$\cdots$} (vn.west);

    \node[sqr] (tn) at ({12.8 * \chainSpread}, 0.90) {$t_n$};
    \node[sqr] (fn) at ({12.8 * \chainSpread}, -0.90) {$f_n$};
    \node[circ] (ak) at ({13.8 * \chainSpread}, 0) {$v'_n$};

    \draw[arr] (vn) -- (tn);
    \draw[arr] (vn) -- (fn);
    \draw[arr] (tn) -- (ak);
    \draw[arr] (fn) -- (ak);

    \draw[arr] (tn) to[bend left=10] node[weight, right] {$1$} (fn);
    \draw[arr] (fn) to[bend left=10] node[weight, left, xshift=1pt] {$k - 1$} (tn);

    \draw[arr] (ak) to[bend right=35] node[weight, above right] {$k - 1$} (tn);
    \draw[arr] (ak) to[bend left=35] node[weight, below right] {$1$} (fn);

    \node[anchor=north, yshift=-3mm] at (vn.south) {$x_n$};

    \node[sqr] (C) at \cPosition {$c$};

    \node[circ] (C1) at ([xshift=-\clauseSpread cm, yshift=-1.10cm]C) {$C_1$};

    \node[circ] (litx1) at ([xshift=-\literalSpread cm, yshift=-1cm]C1) {$x_1$};
    \node[circ] (litx2) at ([yshift=-1cm]C1) {$x_2$};
    \node[circ] (litnx3) at ([xshift=\literalSpread cm, yshift=-1cm]C1) {$\bar{x}_3$};

    \draw[arr] (C) -- (C1);
    \draw[arr] (C1) -- (litx1);
    \draw[arr] (C1) -- (litx2);
    \draw[arr] (C1) -- (litnx3);

    \node[circ] (C2) at ([yshift=-1.10cm]C) {$C_2$};
    \node[circ] (litnx1) at ([xshift=-\literalSpread cm, yshift=-1cm]C2) {$\bar{x}_1$};
    \node[circ] (litnx2) at ([yshift=-1cm]C2) {$\bar{x}_2$};
    \node[circ] (litc23) at ([xshift=\literalSpread cm, yshift=-1cm]C2) {};

    \draw[arr] (C) -- (C2);
    \draw[arr] (C2) -- (litnx1);
    \draw[arr] (C2) -- (litnx2);
    \draw[arr] (C2) -- (litc23);

    \path (C2.center) -- ([xshift=\clauseSpread cm]C2.center) node[midway] {$\cdots$};

    \node[circ] (Cm) at ([xshift=\clauseSpread cm, yshift=-1.10cm]C) {$C_m$};

    \node[circ] (litcm1) at ([xshift=-\literalSpread cm, yshift=-1cm]Cm) {};
    \node[circ] (litcm2) at ([yshift=-1cm]Cm) {};
    \node[circ] (litcm3) at ([xshift=\literalSpread cm, yshift=-1cm]Cm) {};

    \draw[arr] (C) -- (Cm);
    \draw[arr] (Cm) -- (litcm1);
    \draw[arr] (Cm) -- (litcm2);
    \draw[arr] (Cm) -- (litcm3);

    \draw[arr] (litx1) -- (t1);
    \draw[arr] (litx2) -- (t2);
    \draw[arr] (litnx3) -- (f3);

    \draw[arr] (ak) to[out=90, in=0] (C);
\end{tikzpicture}
    \caption{The mod-$k$ game $\G_\Phi$; unlabelled edges represent paths of length $k$.}
    \label{fig:mod-k-hardness}\label{fig:even-hardness}
\end{figure}
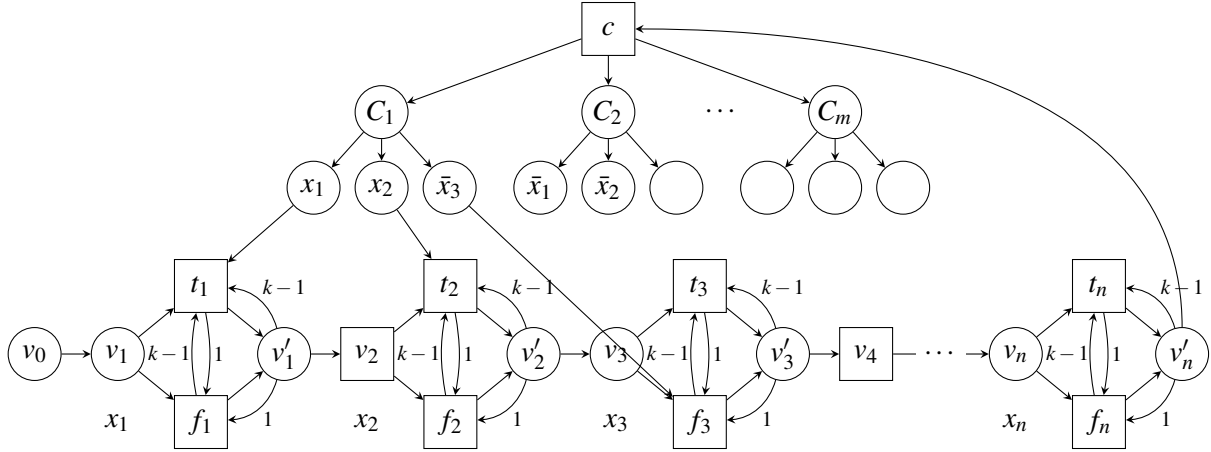

\section{Memory in Lasso Games}\label{sec:memory}
Another way to analyse the complexity of a game is to study the amount of memory needed to play optimally (cf.~\cite{NathBook}). A memory structure is a triple $\mathcal{M} = (M, m_0, \mu)$, where $M$ is a finite set of memory states, $m_0 \in M$ is the initial memory state, and $\mu : E \times M \to M$ is the update function that specifies how the memory changes after a move. We extend the update function to finite sequences of edges $\mu^* : E^* \to M$ as $\mu^*(\epsilon) = m_0$ and $\mu^*(w \cdot e) = \mu(e, \mu^*(w))$. A strategy for Eve using memory $\mathcal{M}$ is a function
\[
    \hat{\sigma} : \Ve \times M \to E.
\]
For every finite play $\play$ ending at vertex $v$, write $\mu^*(\play)$ for the value of $\mu^*$ on the edge sequence of $\play$. The strategy $\hat{\sigma}$ and the memory structure $\mathcal{M}$ then induce the strategy $\sigma(\play) = \hat{\sigma}(v, \mu^*(\play))$. By an abuse of notation, we identify $\hat{\sigma}$ with $\sigma$.

Throughout this section, let $n$ denote the number of vertices in the game arena. For first-cycle games in lasso semantics, memory of size $\mathcal{O}((|\Colors| + 1)^n)$ is always sufficient because one can remember the entire play history. We study when this upper bound can be improved.

\begin{proposition}\label{prop:general-memory-upperbound}
    Let $\Pay : \Colors^* \to \mathbb{R} \cup \{\pm \infty\}$ be a payoff function such that every first-cycle game with payoff $\Pay$ in loop semantics is positionally determined. Then, in every lasso game with payoff $\Pay$ on an arena with $n$ vertices, both players have optimal strategies using at most $2^n$ memory states.
\end{proposition}

\begin{proof}
    The idea is to reduce the game with payoff $\Pay$ in lasso semantics on an arena $\A$ to a game with the same payoff in loop semantics on a larger arena $\A'$. We construct $\A'$ as the product of $\A$ with the power set $\wp(V)$, thereby tracking the set of vertices visited so far. When a vertex would be repeated, the corresponding edge is redirected to the initial vertex $(v_0, \{v_0\})$ rather than closing a cycle locally. Every cycle in $\A'$ therefore passes through $(v_0, \{v_0\})$, yielding a payoff correspondence between lasso plays in $\A$ and cycles in $\A'$.

    Let $\G = (\A, \Pay)$ be a lasso-semantics game, where $\A = (V, E, \Ve, \Va, v_0, \Colors)$ and $|V| = n$. We assume that games with payoff $\Pay$ admit optimal positional strategies for both players in loop semantics. We define the expanded arena $\A' = (V', E', \Ve', \Va', v'_0, \Colors)$ as follows:
    \begin{itemize}
        \item The set of vertices is $V' = \{(v, U) \mid v \in V, \; U \subseteq V, \; v \in U\}$.
        \item The ownership of vertices is inherited from $\A$: $\Ve' = \{(v, U) \in V' \mid v \in \Ve\}$, and $\Va' = \{(v, U) \in V' \mid v \in \Va\}$.
        \item The initial vertex is $v'_0 = (v_0, \{v_0\})$.
        \item For each edge $(v, c, u) \in E$ and each vertex $(v, U) \in V'$, we add an edge to $E'$ as follows:
              \begin{itemize}
                  \item If $u \notin U$, we add $((v, U), c, (u, U \cup \{u\}))$.
                  \item If $u \in U$, we add a \emph{closing} edge $((v, U), c, (v_0, \{v_0\}))$.
              \end{itemize}
    \end{itemize}
    Since $\A$ has no terminal vertices, neither does $\A'$. Note that $|V'| \leq n \cdot 2^n$. Let $\G' = (\A', \Pay)$ be the corresponding game evaluated in loop semantics. Observe that in any play on arena $\A'$, the second coordinate grows until a closing edge is traversed. Thus, the first cycle in $\A'$ is always formed by returning to the initial vertex $(v_0, \{v_0\})$.

    \begin{claim}
        The values of the two games are equal:
        \[
            \val_{\G}(v_0) = \val_{\G'}(v'_0).
        \]
    \end{claim}

    \begin{claimproof}
        We first describe the correspondence between plays in the two games. Consider a nonterminal history $v_0, c_0, v_1, \ldots, c_{k - 1}, v_k$ in $\A$. Since the play has not yet ended, the vertices $v_0, \ldots, v_k$ are pairwise distinct. Let $U_i = \{v_0, \ldots, v_i\}$. The corresponding history in $\A'$ is $(v_0, U_0), c_0, (v_1, U_1), \ldots, c_{k - 1}$, $(v_k, U_k)$. If the next edge in $\A$ is $(v_k, c, u)$, then the corresponding edge in $\A'$ is $((v_k, U_k), c, (u, U_k \cup \{u\}))$ if $u \notin U_k$ and $((v_k, U_k), c, (v_0, \{v_0\}))$ if $u\in U_k$.

        Thus, strategies for either player can be translated between $\G$ and $\G'$ so that corresponding strategy profiles produce plays with the same sequence of labels. Moreover, before a closing edge is taken in $\A'$, the second component strictly grows. Hence, no vertex of $\A'$ is repeated before the closing edge returns to the initial vertex $(v_0, \{v_0\})$. Therefore, the sequence of labels evaluated in loop semantics in $\G'$ is exactly the sequence evaluated in lasso semantics in $\G$, and corresponding strategy profiles have the same payoff.

        Since strategies can be translated in both directions while preserving the payoff of every resulting play, taking the maximum over Eve's strategies and the minimum over Adam's strategies yields
        \[
            \val_{\G}(v_0) = \val_{\G'}(v'_0).
        \]
    \end{claimproof}

    By assumption, both players have optimal positional strategies in the loop-semantics game $\G'$. Each such positional strategy induces a memory-based strategy in the original game $\G$ using a finite memory structure $\mathcal{M} = (M, m_0, \mu)$ with $M = \wp(V)$, initial state $m_0 = \{v_0\}$, and update function $\mu((u, c, v), U) = U \cup \{v\}$. The next edge is obtained by taking in $\A$ an edge corresponding to the edge prescribed in $\A'$ at state $(v, U)$. Since the games have the same value and the positional strategies in $\G'$ are optimal, the induced strategies in $\G$ are also optimal. The memory structure has size $|M| = 2^n$.
\end{proof}

\subsection{Mean-payoff games}

\begin{restatable}{proposition}{propMPMemoryUpper}\label{prop:lasso-MP-memory-upper}
    In any lasso mean-payoff game played on an arena with $n$ vertices, a memory structure of size $2^n$ is sufficient for both players to play optimally.
\end{restatable}

\begin{proof}
    As established by Ehrenfeucht and Mycielski~\cite{Ehrenfeucht1979Positional}, mean-payoff games are positionally determined in loop semantics. Consequently, the $2^n$ memory upper bound follows from~\Cref{prop:general-memory-upperbound}.
\end{proof}

\begin{restatable}{proposition}{propMPMemoryLower}\label{prop:lasso-MP-memory-lower}
    There exists a family of lasso mean-payoff games in which Eve requires exponential memory to play optimally.
\end{restatable}

\begin{proof}
    For every $r \geq 1$, we construct a lasso mean-payoff game $\G_r$ with $n = 4r + 1$ vertices in which every optimal strategy for Eve requires at least $2^r$ memory states. The game $\G_r$ consists of $r$ connected gadgets corresponding to the variables $v_1, \dots, v_r$. As before, the game has two phases. In the \emph{assignment phase}, Adam assigns a value ($L$ or $R$) to each variable. In the subsequent \emph{query phase}, Eve is tested on one of Adam's choices. The game $\G_3$ is shown in~\Cref{fig:MP-memory-lower}.

    \begin{figure}[hbt]
        \centering
        \begin{tikzpicture}[
        circ/.style={circle, draw, minimum size=7mm, inner sep=0pt, fill=white},
        sqr/.style={rectangle, draw, minimum size=7mm, inner sep=0pt, fill=white},
        arr/.style={->, >=stealth, thick},
        line/.style={thick}, weight/.style={font=\scriptsize}
    ]
    \def\hsep{3.2}

    \foreach \i [count=\c from 0] in {1, 2, 3} {
            \begin{scope}[xshift=\c * \hsep cm]
                \node[sqr] (v\i) at (0, 0) {$v_{\i}$};
                \node[sqr] (l\i) at (1.6, -1) {$v_{\i, L}$};
                \node[sqr] (r\i) at (1.6, 1) {$v_{\i, R}$};
                \node[circ] (q\i) at (0, 2) {$q_{\i}$};
                \draw[arr] (v\i) -- node[weight, below left] {$1$} (l\i);
                \draw[arr] (v\i) -- node[weight, above left] {$1$} (r\i);
                \draw[arr] (q\i) -- (l\i);
                \draw[arr] (q\i) -- (r\i);
            \end{scope}
        }

    \path (v3) ++(3.2, 0) node[sqr] (c) {$q$};

    \draw[arr] (l1) -- node[weight, below] {$-1$} (v2);
    \draw[arr] (r1) -- node[weight, above] {$-1$} (v2);
    \draw[arr] (l2) -- node[weight, below] {$-1$} (v3);
    \draw[arr] (r2) -- node[weight, above] {$-1$} (v3);
    \draw[arr] (l3) -- node[weight, below] {$-1$} (c);
    \draw[arr] (r3) -- node[weight, above] {$-1$} (c);

    \coordinate (busTop) at (0, 2.8);
    \draw[line] (c.north) -- (c |- busTop) -- (q1 |- busTop);
    \foreach \i in {1, 2, 3} {\draw[arr] (q\i |- busTop) -- (q\i);}
\end{tikzpicture}
        \caption{The game $\G_3$, with initial vertex $v_1$. Unlabelled edges have weight $0$.}
        \label{fig:MP-memory-lower}
    \end{figure}
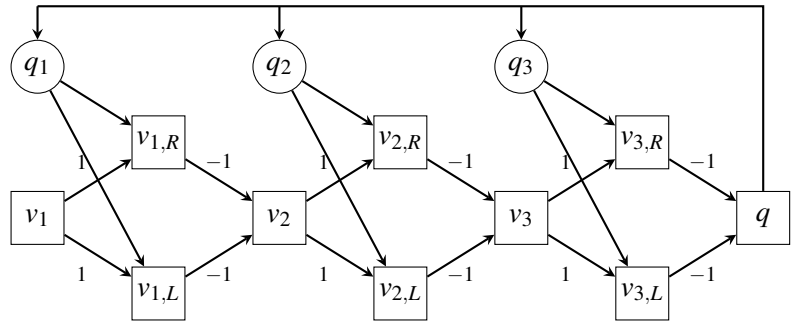

    Eve has an optimal strategy $\strE$ using $2^r$ memory states. She simply keeps track of Adam's choices in every gadget. In the query phase, Adam asks Eve about the value assigned to some variable $v_i$. Since she remembers the entire assignment, she takes the outgoing edge to $v_{i, d}$, where $d \in \{L, R\}$ is the value selected by Adam during the assignment phase. The play ends with mean payoff $0$.

    Now suppose that Eve follows a strategy $\strE$ using fewer than $2^r$ memory states. Then two distinct plays $\play$ and $\play'$, reach vertex $q$ in the same memory state. Since the memory state and the current vertex determine Eve's next edge, both plays continue in the same manner.

    Consequently, Adam can choose an index $i$, $1 \leq i \leq r$, on which the two assignments differ: in one play he chose $v_{i, L}$, and in the other he chose $v_{i, R}$. Eve chooses the same outgoing edge in both plays, so in one of them she chooses an edge leading to an unvisited vertex. The payoff of at least one of the plays $\play$ or $\play'$ is therefore less than $0$, so $\strE$ is not optimal.
\end{proof}

\subsection{Parity games}

By~\Cref{sec:Parity}, lasso parity games are positionally determined.

\subsection{Mod-\texorpdfstring{$k$}{k} games}

\begin{proposition}\label{prop:mod-memory-upper}\label{prop:even-memory-upper}
    For every $k \geq 2$, in mod-$k$ games in both loop and lasso semantics, memory of size $(k + 1)^n$ is sufficient for both players to play optimally.
\end{proposition}

\begin{proof}
    We first define the function $\Val$ and the strategies induced by it, and then prove their optimality by induction. A state of the game can be represented as a tuple $(v, U_0, \ldots, U_{k - 1})$, where $v \in V$ is the current vertex and the sets $U_0, \ldots, U_{k - 1} \subseteq V$ are pairwise disjoint. The set $U_r$ contains the vertices visited before the current occurrence of $v$ at positions congruent to $r$ modulo $k$. Let $U = \bigcup_{r = 0}^{k - 1} U_r$. Before the first repetition, all visited vertices are distinct, so $|U|$ is exactly the number of edges traversed so far. We let $\Val(v, U_0, \ldots, U_{k - 1})$ be true if Eve can force a win from this state, and false otherwise.

    When $v \in U$, the play has just ended. We define the terminal value $T(v, U_0, \ldots, U_{k - 1})$ as follows. In lasso semantics, $T(v, U_0, \ldots, U_{k - 1}) = \bigl(|U| \equiv 0 \pmod{k}\bigr)$. In loop semantics, let $r$ be the unique residue such that $v \in U_r$. The previous occurrence of $v$ was at a position congruent to $r$ modulo $k$, so the loop has length congruent to $|U| - r$ modulo $k$. Hence, $T(v, U_0, \ldots, U_{k - 1}) = \bigl(|U| - r \equiv 0 \pmod{k}\bigr)$. For either semantics, we define $\Val$ recursively as
    \[
        \Val(v, U_0, \ldots, U_{k - 1}) =
        \begin{cases}
            T(v, U_0, \ldots, U_{k - 1}),                                                  & v \in U,                     \\[2mm]

            \displaystyle\bigvee_{(v, c, u) \in E(v)} \Val(u, U_0', \ldots, U_{k - 1}'),   & v \notin U, \quad v \in \Ve, \\[3mm]

            \displaystyle\bigwedge_{(v, c, u) \in E(v)} \Val(u, U_0', \ldots, U_{k - 1}'), & v \notin U, \quad v \in \Va,
        \end{cases}
    \]
    where $(U_0', \ldots, U_{k - 1}')$ is obtained by adding $v$ to
    $U_{|U| \bmod k}$.

    Let the memory structure consist of all tuples $(U_0, \ldots, U_{k - 1})$, with all sets initially empty. After traversing an edge $(v, c, u)$, the memory is updated by adding $v$ to $U_{|U| \bmod k}$. At an Eve state with value true, Eve chooses an edge to a state that also evaluates to true. At an Adam state with value false, Adam chooses an outgoing edge to a state that also evaluates to false. In all other states, the choice can be arbitrary.

    We prove that these strategies are winning by induction on the number of vertices not in $U$. If $v \in U$, the play has ended and the terminal value agrees with the corresponding mod-$k$ objective. Otherwise, every edge leads to a state with one additional vertex in $U$, so the induction hypothesis applies. At an Eve vertex, a true disjunction provides an edge to a state from which Eve wins, while a false disjunction means that every outgoing edge is losing for her. At an Adam vertex, a true conjunction means that every outgoing edge is winning for Eve, while a false conjunction provides an edge to a state from which Adam wins. Thus, $\Val$ correctly describes the winner from every state, and the strategies defined above are winning whenever the corresponding player can win.

    Finally, each vertex is either unvisited or belongs to exactly one of the $k$ sets $U_0, \ldots, U_{k - 1}$. Hence $(k + 1)^n$ memory states are sufficient.
\end{proof}

\begin{restatable}{proposition}{propEvenMemoryLower}\label{prop:mod-memory-lower}\label{prop:even-memory-lower}\label{prop:lasso-mod-memory-lower}
    For every fixed $k \geq 2$, there exists a family of mod-$k$ games, in both loop and lasso semantics, with $n$ vertices in which Eve requires $2^{\Omega(n)}$ memory states to win.
\end{restatable}

\begin{proof}
    Fix $k \geq 2$. For every $r \geq 1$, we construct a mod-$k$ game $\G_r$ from $r$ connected gadgets corresponding to the variables $v_1, \dots, v_r$. As before, the game has an \emph{assignment phase}, in which Adam assigns $L$ or $R$ to every variable, followed by a \emph{query phase}, in which Eve is tested on one of his choices. We use the gadgets from~\Cref{fig:mod-k-gadget}. The game $\G_3$ is shown in~\Cref{fig:mod-memory-lower}.

    \begin{figure}[hbt!]
        \centering
        \begin{tikzpicture}[
        circ/.style={circle, draw, minimum size=7mm, inner sep=0pt, fill=white},
        sqr/.style={rectangle, draw, minimum size=7mm, inner sep=0pt, fill=white},
        arr/.style={->, >=stealth, thick},
        line/.style={thick},
        weight/.style={font=\scriptsize}
    ]
    \def\hsep{4.8}

    \foreach \i [count=\c from 0] in {1, 2, 3} {
            \begin{scope}[xshift=\c * \hsep cm]
                \node[sqr] (v\i) at (0, 0) {$v_{\i}$};
                \node[sqr] (l\i) at (1.6, -1) {$v_{\i, L}$};
                \node[sqr] (r\i) at (1.6, 1) {$v_{\i, R}$};
                \node[circ] (a\i) at (3.2, 0) {$v'_{\i}$};
                \node[circ] (q\i) at (0, 2) {$q_{\i}$};

                \draw[arr] (v\i) -- (l\i);
                \draw[arr] (v\i) -- (r\i);
                \draw[arr] (l\i) -- (a\i);
                \draw[arr] (r\i) -- (a\i);

                \draw[arr] (l\i) to[bend left=10] node[weight, left] {$1$} (r\i);
                \draw[arr] (r\i) to[bend left=10] node[weight, right] {$k - 1$} (l\i);

                \draw[arr] (a\i) to[bend left=50] node[weight, right] {$k - 1$} (l\i);
                \draw[arr] (a\i) to[bend right=50] node[weight, right] {$1$} (r\i);

                \draw[arr] (q\i) -- (r\i);
                \draw[arr] (q\i) -- (l\i);
            \end{scope}
        }
    \node[sqr] (adam) at ([xshift=1.5cm]a3) {$q$};
    \coordinate (busLine) at (0, 2.75);

    \draw[arr] (a1) -- (v2);
    \draw[arr] (a2) -- (v3);
    \draw[arr] (a3) -- (adam);

    \draw[line] (adam.north) -- (adam.north |- busLine) -- (q1 |- busLine);

    \foreach \i in {1, 2, 3} {\draw[arr] (q\i |- busLine) -- (q\i);}
\end{tikzpicture}
        \caption{The game $\G_3$, with initial vertex $v_1$. Unlabelled edges represent paths of length $k$.}
        \label{fig:mod-memory-lower}
        \label{fig:even-memory-lower}
    \end{figure}
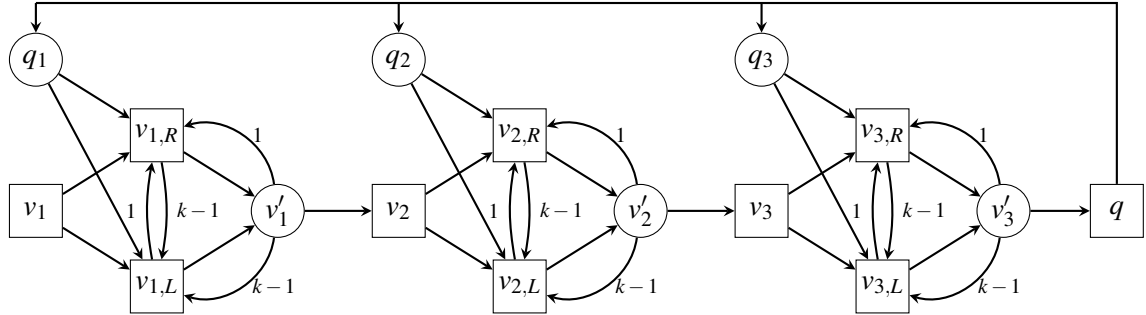

    Eve has a winning strategy $\strE$ using $2^r$ memory states. She keeps track of Adam's choices in every gadget. Any deviation by Adam can be punished as in~\Cref{clm:mod-k-gadget}. In the query phase, Adam selects a vertex $q_i$, and Eve takes the outgoing edge to the previously visited vertex $v_{i, d}$, where $d \in \{L, R\}$ is the value selected by Adam during the assignment phase. Both the resulting loop and the lasso have length divisible by $k$.

    Now suppose that Eve follows a strategy $\strE$ using fewer than $2^r$ memory states. Then two distinct assignments reach the query phase in the same memory state. Choose an index $i$ on which they differ, and let Adam select $q_i$. Eve chooses the same outgoing edge in both plays, so in one of them she chooses an edge leading to an unvisited vertex. Adam then takes the cross-path to the previously visited opposite vertex. The resulting loop has length $1$ or $k - 1$ modulo $k$. The head has length divisible by $k$, so the lasso has the same non-zero residue. Thus Eve loses in both semantics, contradicting the assumption that $\strE$ is winning.

    The arena has $n = \mathcal{O}(kr)$ vertices. Therefore, every winning strategy requires at least $2^r = 2^{\Omega(n / k)} = 2^{\Omega(n)}$ memory states for every fixed $k$.
\end{proof}

\subsection*{Acknowledgment}
This work was supported by the National Science Centre, Poland (grant no.~2021/41/B/ST6/03914).

\bibliographystyle{eptcs}
\bibliography{bibliography}

\end{document}